\documentclass[letterpaper, 10 pt, conference]{ieeeconf}  
\IEEEoverridecommandlockouts                              
\usepackage{graphics} 
\usepackage{epsfig} 
\usepackage{amsmath, amssymb, ntheorem, amsfonts}
\usepackage{algorithm}
\usepackage{array}
\usepackage{algpseudocode}

\usepackage[pdfpagelabels,  
bookmarks,       
pdftex
]{hyperref}
\hypersetup{
	pdfborder   = 0 0 0,
	plainpages  = false,
	bookmarksnumbered = true,
}
\usepackage[nolist]{acronym}
\usepackage{balance}
\usepackage{todonotes}
\usepackage{graphicx,import}
\def\svgwidth{\linewidth}

\theoremstyle{plain} 
\theoremheaderfont{\normalfont\bfseries} 
\theorembodyfont{\normalfont} 

\newtheorem{theorem}{Theorem} 
\newtheorem{lemma}{Lemma} 

\newtheorem{remark}{Remark}

\begin{acronym}
	\acro{rnn}[RNN]{Recurrent Neural Network}
	\acroplural{rnn}[RNN]{Recurrent Neural Networks}
	\acro{narx}[NARX-NN]{Nonlinear Autoregressive Exogenous Neural Network}
	\acroplural{narx}[NARX-NN]{Nonlinear Autoregressive Exogenous Neural Networks}
	\acro{gru}[GRU]{{Gated Recurrent Unit}}
	\acroplural{gru}[GRU]{{Gated Recurrent Units}}
	\acro{lstm}[LSTM]{Long Short-Term Memory}
	\acro{ann}[NN]{Neural Network}	
	\acroplural{ann}[NN]{Neural Networks}
	\acro{ffnn}[FFNN]{Feedforward Neural Network}
	\acroplural{ffnn}[FFNN]{Feedforward Neural Networks}
	\acro{pinn}[PINN]{Physics-Informed Neural Network}
	\acroplural{pinn}[PINN]{Physics-Informed Neural Networks}
	\acro{gp}[GP]{Gaussian Process}
	\acroplural{gp}[GP]{Gaussian Processes}
	\acro{knn}[$K$NN]{$K$-Nearest-Neighbors}
	\acro{ilc}[ILC]{Iterative Learning Control}
	\acro{rc}[RC]{Repetitive Control}
	\acro{rl}[RL]{{Reinforcement Learning}}
	\acro{daoc}[DAOC]{{Direct Adaptive Optimal Control}}
	\acro{ml}[ML]{maschinelles Lernen}
	\acro{lwpr}[LWPR]{{Locally Weighted Projection Regression}}
	\acro{svm}[SVM]{{Support Vector Machine}}
	\acro{mcmc}[MCMC]{{Markov Chain Monte-Carlo}}
	\acro{ad}[AD]{{Automatic Differentiation}}
	\acro{gmm}[GMM]{{Gaussian Mixture Models}}
	\acro{rkhs}[RKHS]{{Reproducing Kernel Hilbert Spaces}}
	\acro{rbf}[RBF]{{Radial Basis Function}}
	\acro{rbfnn}[RBF-NN]{{Radial Basis Function Neural Network}}
	\acroplural{rbfnn}[RBF-NN]{{Radial Basis Function Neural Networks}}
	\acro{node}[{Neural} ODE]{{Neural Ordinary Differential Equation}}
	\acroplural{node}[{Neural} ODEs]{{Neural Ordinary Differential Equations}}
	\acro{pdf}[PDF]{Probability Density Function}
	\acro{pca}[PCA]{Principal Component Analysis}

	\acro{kf}[KF]{Kalman Filter}	
	\acro{ekf}[EKF]{Extended Kalman Filter}
	\acro{nekf}[NEKF]{Neural Extended Kalman Filter}
	\acro{ukf}[UKF]{Unscented Kalman Filter}
	\acro{pf}[PF]{Particle Filter}
	\acro{mhe}[MHE]{{Moving Horizon Estimation}}
	\acro{rpe}[RPE]{{Recursive Predictive Error}}
	\acro{rls}[RLS]{{Recursive Least Squares}}
	\acro{slam}[SLAM]{{Simultaneous Location and Mapping}}
	
	\acro{mftm}[MFTM]{Magic Formula Tire Model}
	\acro{mbs}[MBS]{Multi-Body Simulation}
	\acro{lti}[LTI]{Linear Time-Invariant}
	\acro{cog}[COG]{Center Of Gravity}
	\acro{ltv}[LTV]{Linear Time-Variant}
	\acro{siso}[SISO]{Single Input Single Output}
	\acro{mimo}[MIMO]{Multiple Input Multiple Output}
	\acro{psd}[PSD]{Power Spectral Density}
	\acroplural{psd}[PSD]{Power Spectral Densities}
	\acro{cf}[CF]{Coordinate Frame}
	\acroplural{cf}[CF]{Coordinate Frames}
	\acro{bf}[basis function]{basis function}
	\acroplural{bf}[basis functions]{basis functions}
	
	\acro{pso}[PSO]{Particle Swarm Optimization}
	\acro{sqp}[SQP]{Sequentielle Quadratische Programmierung}
	\acro{svd}[SVD]{Singular Value Decomposition}
	\acro{ode}[ODE]{Ordinary Differential Equation}
	\acroplural{ode}[ODE]{Ordinary Differential Equations}
	\acro{pde}[PDE]{Partial Differential Equation}
	
	\acro{nmse}[NMSE]{Normalized Mean Squared Error}
	\acro{mse}[MSE]{Mean Squared Error}
	\acro{rmse}[RMSE]{Root Mean Squared Error}
	
	\acro{mpc}[MPC]{{Model Predictive Control}}
	\acro{nmpc}[NMPC]{Nonlinear Model Predictive Control}
	\acro{lmpc}[LMPC]{Learning Model Predictive Control}
	\acro{ltvmpc}[LTV-MPC]{\acl{ltv} Model Predictive Control}
	\acro{ndi}[NDI]{Nonlinear Dynamic Inversion}
	\acro{ac}[AC]{Adhesion Control}
	\acro{esc}[ESC]{Electronic Stability Control}
	\acro{ass}[ASS]{Active Suspension System}
	\acro{trc}[TRC]{Traction Control}
	\acro{abs}[ABS]{Anti-Lock Brake System}
	\acro{gcc}[GCC]{Global Chassis Control}
	\acro{ebs}[EBS]{Electronic Braking System}
	\acro{adas}[ADAS]{Advanced Driver Assistance Systems}
	\acro{cm}[CM]{{Condition Monitoring}}
	\acro{hil}[HiL]{{Hardware-in-the-Loop}}
	\acro{siso}[SISO]{Single Input Single Output}
	\acro{mimo}[MIMO]{Multiple Input Multiple Output}

	\acro{irw}[IRW]{Independently Rotating Wheels}
	\acro{dirw}[DIRW]{Driven \acl{irw}}
	\acro{imes}[imes]{{Institute of Mechatronic Systems}}
	\acro{db}[DB]{Deutsche Bahn}
	\acro{ice}[ICE]{Intercity-Express}
	
	\acro{fmi}[FMI]{{Functional Mock-up Interface}}
	\acro{fmu}[FMU]{{Functional Mock-up Unit}}
	\acro{doi}[DOI]{{Digital Object Identifier}}

	\acro{ra}[RA]{Research Area}
	\acroplural{ra}[RA]{Research Areas}
	\acro{wp}[WP]{Work Package}
	\acroplural{wp}[WP]{Work Packages}
	
	\acro{fb}[FB]{Forschungsbereich}
	\acroplural{fb}[FB]{Forschungsbereiche}
	\acro{ap}[AP]{Arbeitspaket}
	\acroplural{ap}[AP]{Arbeitspakete}
	\acro{abb}[Abb.]{Abbildung}
	
	\acro{res}[RES]{Renewable Energy Sources}
	\acro{pkw}[PKW]{Personenkraftwagen}
	\acro{twipr}[TWIPR]{{Two-Wheeled Inverted Pendulum Robot}}
	
	\acro{pmcmc}[PMCMC]{Particle Markov Chain Monte-Carlo}
	\acro{mcmc}[MCMC]{Markov Chain Monte-Carlo}
	\acro{rbpf}[RBPF]{Rao-Blackwellized Particle Filter}
	\acro{pf}[PF]{Particle Filter}
	\acro{ps}[PS]{Particle Smoother}
	\acro{smc}[SMC]{Sequential Monte-Carlo}
	\acro{mh}[MH]{Metropolis Hastings}
	\acro{em}[EM]{Expectation Maximization}
	\acro{slam}[SLAM]{Simultaneous Location and Mapping}
	\acro{dof}[DOF]{Degree of Freedom}
	\acroplural{dof}[DOF]{Degrees of Freedom}

\end{acronym}

\newcommand{\eg}{e.\,g.,\,}
\newcommand{\ie}{i.\,e.,\,}
\newcommand*{\tr}{^{\top}}
\newcommand*{\R}{\mathbb{R}}

\newcommand*{\IW}{\mathcal{IW}}
\newcommand*{\T}{\mathcal{T}}
\newcommand*{\N}{\mathcal{N}}

\newcommand{\bi}[1]{\boldsymbol{#1}}

\newcommand\norm[1]{\left\lVert#1\right\rVert}

\DeclareMathOperator*{\argmin}{arg\,min}
\DeclareMathOperator{\prob}{\mathit{p}\!}
\DeclareMathOperator{\spanop}{span\!}
\newcommand{\spans}[1]{\spanop\left( #1 \right)}

\newcommand\copyrighttext{%
	\footnotesize Accepted for presentation at 2024 Conf. on Decision and Control. \copyright 2024 IEEE.  Personal use of this material is permitted. Permission from IEEE must be obtained for all other uses, in any current or future media, including reprinting/republishing this material for advertising or promotional purposes, creating new collective works, for resale or redistribution to servers or lists, or reuse of any copyrighted component of this work in other works.}

\newcommand\copyrightnotice{%
	\begin{tikzpicture}[remember picture,overlay]
		\node[anchor=north,yshift=-5mm] at (current page.north) {\fbox{\parbox{\dimexpr\textwidth-\fboxsep-\fboxrule\relax}{\copyrighttext}}};
	\end{tikzpicture}%
}

\title{\LARGE \bf
	Efficient Online Inference and Learning in Partially Known Nonlinear State-Space Models by Learning Expressive Degrees of Freedom Offline
}

\author{Jan-Hendrik Ewering, Bj\"orn Volkmann, Simon F.\,G. Ehlers, Thomas Seel, and Michael Meindl
	\thanks{The Authors are with the Institute of Mechatronic Systems, Leibniz Universit\"at Hannover, 30167~Hanover, Germany (e-mail: \href{mailto:jan-hendrik.ewering@imes.uni-hannover.de}{{\tt\small jan-hendrik.ewering@imes.uni-hannover.de}}).}%
}

\begin{document}

	\maketitle\copyrightnotice\vspace{-9.5pt}
	\thispagestyle{empty}
	\pagestyle{empty}

	\begin{abstract}
		
		Intelligent real-world systems critically depend on expressive information about their system state and changing operation conditions, \eg due to variation in temperature, location, wear, or aging. 
		To provide this information, online inference and learning attempts to perform state estimation and (partial) system identification simultaneously. Current works combine tailored estimation schemes with flexible learning-based models but suffer from convergence problems and computational complexity due to many degrees of freedom in the inference problem (\ie parameters to determine). To resolve these issues, we propose a procedure for data-driven offline conditioning of a highly flexible \ac{gp} formulation such that \textit{online} learning is restricted to a subspace, spanned by expressive basis functions.  
		Due to the simplicity of the transformed problem, a standard particle filter can be employed for Bayesian inference. In contrast to most existing works, the proposed method enables online learning of target functions that are nested nonlinearly inside a first-principles model. Moreover, we provide a theoretical quantification of the error, introduced by restricting learning to a subspace. 
		A Monte-Carlo simulation study with a nonlinear battery model shows that the proposed approach enables rapid convergence with significantly fewer particles compared to a baseline and a state-of-the-art method.
	\end{abstract}
	
%
	
	\acresetall
	\section{INTRODUCTION}
	Operation under complex and changing conditions is a key challenge in modern control research. The changing conditions can be attributed to intrinsic system behavior (\eg wear and friction in machines \cite{Ishiyama.2023}, aging of batteries \cite{Aitio.2023}) or environment interaction (\eg unknown environment map \cite{Kok.2024,Viset.2022}, changing tire-road friction \cite{Berntorp.2021,Berntorp.2022,Lampe.2023}), see Fig.~\ref{fig:BasisFunctionComparisson}. In both cases, information on the underlying change is crucial to ensure adaptive and reliable operation. In this light, estimation algorithms fuse assumptions about the system structure with (limited) sensor data to obtain estimates of (latent) system states and varying parameters.\\
	Usually, an approximate system model can be derived from first principles, such as rigid body dynamics. However, knowledge about other aspects influencing the system is often limited, \eg friction effects, or environment maps.\\
	To address this, offline algorithms for joint inference of latent system states and learning of (partially) unknown models have been proposed \cite{Svensson.2016,Svensson.2017,BuissonFenet.2023b}. In \cite{Svensson.2016,Svensson.2017}, for instance, learning is facilitated through encoding basic assumptions about smoothness and dimensionality of the underlying true system behavior using a \ac{gp} prior \cite{Solin.2020}. For inference and learning, the authors employ particle Markov chain Monte Carlo methods \cite{Andrieu.2010}. However, such algorithms are often computationally demanding and rely on offline data to train highly flexible learning-based models, which hinders real-time adaptivity to changing conditions.\\
	In contrast, recent estimation algorithms attempt to learn (partially) unknown system behavior online while simultaneously inferring latent states. In \cite{Berntorp.2021}, Bayesian online learning of \ac{gp} state-space models is proposed based on \cite{Svensson.2016,Svensson.2017} and using a carefully designed \ac{smc} algorithm. However, high-dimensional search spaces can still render the estimation problem infeasible due to (i) convergence issues associated with complex posterior probability densities, and (ii) computational complexity. The associated challenges are commonly termed ``curse of dimensionality'', especially in the context of \ac{smc} methods.
	\begin{figure}[t]
		\centering
		\fontsize{8pt}{10pt}\selectfont 
		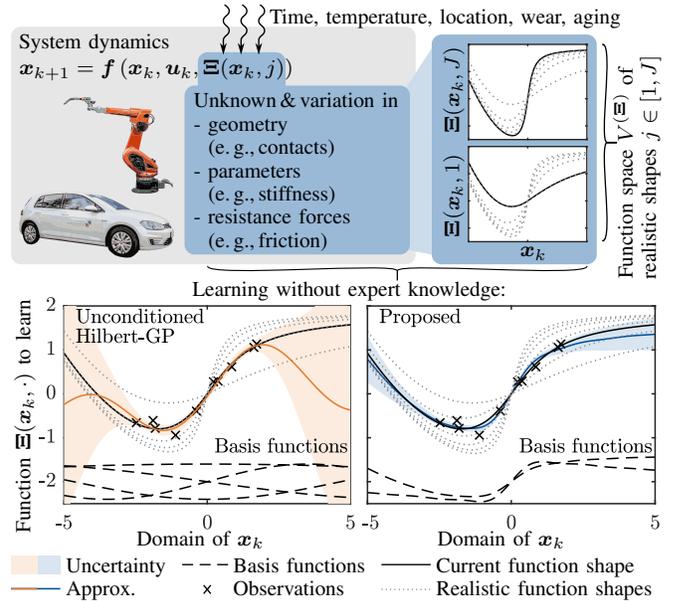
		\normalsize 
		\caption{Unknown and varying effects $\bi{\Xi}$ in real-world systems that complicate operation (top) and approaches for learning of the underlying relationships (bottom). Without expert knowledge, online learning is difficult with state-of-the-art methods, \eg Hilbert-\ac{gp} \cite{Solin.2020}, due to many inexpressive degrees of freedom (here: basis functions). In contrast, the proposed method enables efficient online learning by \textit{data-driven} construction of few expressive basis functions.}
		\label{fig:BasisFunctionComparisson}
	\end{figure}\\
	To resolve these issues, a key approach is restricting learning to expressive \acp{dof}. In particular, what we mean with ``expressive \ac{dof}'' is that the model structure should exhibit few adjustable parameters, each of which has a unique and significant impact on the considered target, \eg the function shape to be learned.\\
	In estimation settings with purely physics-based models, this would be the derivation of empirical expert models in which only few parameters need to be determined. A common example is to represent the tire-road friction using the magic formula tire model \cite{Pacejka.1992} and to employ it for automotive estimation, \eg using unscented Kalman filters \cite{Lampe.2023}.\\
	On the other hand, learning-based models can be constrained to facilitate adaptivity in online settings. In \cite{Berntorp.2022}, expert knowledge is exploited in the form of symmetry and linear operator constraints to restrict flexibility. While the approaches in \cite{Berntorp.2022} are important contributions to simplify the learning problem, they are still limited to specific cases, and expert knowledge may not be available. The authors of \cite{Kullberg.2021c} solve a high-dimensional inference task by restricting online learning to a region around the current operation point. However, the function shape itself is not constrained, resulting in erroneous target function estimates. A notable contribution for learning of a low-dimensional, yet flexible \ac{gp} model is presented in \cite{Menzen.2023}. The authors connect a tractable tensor network with a linear-in-the-parameters \ac{gp} formulation \cite{Solin.2020} to enable inference in a low-dimensional subspace, but the advantages are not exploited for online inference and learning.\\
	In the vast majority of previous publications, the target function is learned either with overly flexible models (\ie too many \acp{dof}) or employing system-specific expert knowledge. Unlike all prior work, we propose to learn the most significant features of observed function realizations from data offline and to use them as expressive \acp{dof} for efficient Bayesian online inference and learning. By doing this, we restrict online learning to a low-dimensional subspace that spans only realistic function shapes, yielding fast convergence 
	and reduced computational burden,
	while employing a standard \ac{pf}. Due to the simplicity of the inference problem and in contrast to most previous work, the proposed methods are capable of learning functions online that are nested inside nonlinear first-principles models, without expert knowledge about the target functions.\\ 
	The paper is structured as follows. First, the problem is formalized in Sec.~\ref{sec:problem}. The proposed method for efficient online inference and learning is presented in Sec.~\ref{sec:methods}. Last, we illustrate our findings with a simulation example and draw conclusions in Sec.~\ref{sec:results} and \ref{sec:conclusion}, respectively.
	
	\textit{Notation}: For a vector $\bi{e} \in \R^{n_e}$, $\bi{e} \sim \N\left(\bi{\mu},\bi{\Sigma}\right)$ denotes a draw from a multivariate Normal with mean $\bi{\mu}$ and covariance $\bi{\Sigma}$ and $e_i$ is its $i$-th element. We use column vectors if not stated explicitly otherwise. A matrix $\bi{A}$ is written in bold and capital and has elements $a_{ij}$ for row $i$ and column $j$. The identity matrix of dimension $n$ is $\mathbf{I}_n$. We write the conditional density of a state sequence $\bi{x}_{1:k} := \{ \bi{x}_i \}_{i=1}^k$ from time steps $1$ to $k$, given the measurements $\bi{y}_{1:k}$ as $\prob(\bi{x}_{1:k}|\bi{y}_{1:k})$. By writing $V^{(\phi,N)} := \spans{\phi_1(\bi{x}),\hdots,\phi_N(\bi{x})}$, we refer to the space of functions that can be represented by linear combinations of basis functions  $\{\phi_i(\bi{x})\}_{i=1}^{N}$, defined on the domain $\bi{x} \in \Omega$. The inverse Wishart distribution with scale matrix $\bi{\Lambda}$ and \ac{dof} $\nu$ is $\IW(\nu,\bi{\Lambda})$. The multivariate Student-t distribution is $\T\left(\nu, \bi{\mu},\bi{\Lambda}\right)$. The Dirac delta mass $\delta_{ij} = 1$ for $i=j$ and $0$ otherwise. By writing $\norm{\bi{x}}_{\bi{M}}^2$, we mean $\bi{x}\tr \bi{M} \bi{x}$ and $\norm{\bi{x}}$ denotes the L$2$-norm of $\bi{x}$. 
	
	\section{PROBLEM FORMULATION}\label{sec:problem}
	The objective is to learn the nested and changing system behavior \textit{online} from noisy input-output data while using first-principles model knowledge. 
	In a probabilistic discrete-time state-space model
	\begin{subequations}\label{eq:model_structure}
		\begin{align}
			{\bi{x}}_{k+1} &  = \overbrace{\bi{f} \left(\bi{x}_{k}, \bi{u}_{k}, \bi{\Xi}(\bi{x}_{k}, j) \right)}^{\substack{ \text{\tiny Dominating first-principles model} \\ \text{\tiny \& nested unknown effects to be learned}}}  +  \bi{e}_{k}^{x}, \\
			\bi{y}_{k} & =\bi{{h}}(\bi{x}_{k},\bi{u}_{k}) + \bi{e}_{k}^{y},
		\end{align}
	\end{subequations}
	this amounts to learning the function $\bi{\Xi}:\R^{n_x + 1} \rightarrow \R^{n_{\xi}}$ and inferring the latent states $\bi{x}_k \in \R^{n_x}$ at time step $k$. The variation of $\bi{\Xi}(\cdot,j)$ in a set of typical shapes, \ie the function space $V^{(\bi{\Xi})}$, is described by a scheduling variable $j \in \left[1, J\right]$ (see Fig.~\ref{fig:BasisFunctionComparisson}). Please note that the scheduling variable $j$ is introduced for notational convenience and is not assumed to be known for online inference and learning. The case of $\bi{\Xi}$ being nested in $\bi{h}$ is conceptually similar and will not be considered explicitly. In \eqref{eq:model_structure}, the inputs are $\bi{u}_k \in \R^{n_u}$, and the measurements $\bi{y}_k \in \R^{n_y}$. The process noise $\bi{e}_k^{x}$ and the measurement noise $\bi{e}_k^{y}$ are zero-mean Gaussian random variables with known covariance matrices $\bi{Q}$ and $\bi{R}$, \ie $\bi{e}_k^{x} \sim \N\left(\bi{0},\bi{Q}\right)$ and $\bi{e}_k^{y} \sim \N\left(\bi{0},\bi{R}\right)$. The state dynamics $\bi{f}:\R^{n_x} \times \R^{n_u} \times \R^{n_{\xi}} \rightarrow \R^{n_x}$ and the measurement function $\bi{h}:\R^{n_x} \times \R^{n_u} \rightarrow \R^{n_y}$ are known from first-principles.\\
	For learning of $\bi{\Xi}$, a generic approximation model (\eg a \ac{gp}) is required to impose ``artificial structure''. To enable efficient inference and learning, the approximation should (i) simplify the estimation problem, and (ii) be computationally efficient. We assume that we do not have system-specific expert knowledge about the function structure of $\bi{\Xi}$ and its variation in $V^{(\bi{\Xi})}$. However, we are given an \textit{offline} data set $\mathcal{D} = \{\bi{\xi}_{1:K}^{j},\bi{x}_{1:K}^{j} \}_{j=1}^{J}$ with noisy observations $\bi{\xi}_k = \bi{\Xi}(\bi{x}_k,j) + \bi{e}_k^{\xi}$, $\bi{e}_k^{\xi} \sim \N (\bi{0}, \sigma_{\xi}\mathbf{I}_{n_\xi})$, that stem from $J$ realizations in the whole range $j \in \left[1, J \right]$. Note that measuring quantities of interest \textit{offline} under laboratory conditions and estimating these quantities \textit{online} in operation is common practice in many applications. Further, established methods could be used to infer $\mathcal{D}$, leveraging a system structure that is affine in $\bi{\Xi}$ \cite{Svensson.2017} or using identification methods for general nonlinear systems \cite{Schon.2021,Wigren.2022}.\\ 
	In the online setting, both state estimation and learning of $\bi{\Xi}$ need to be performed simultaneously in each time step $k$ based on the current inputs $\bi{u}_{k-1}$ and measurements $\bi{y}_k$ only. Apart from convergence, the computational complexity of the algorithm should be as small as possible.
	
	\section{PROPOSED METHOD}\label{sec:methods}
	To enable efficient online inference and learning, the parameters to be learned online should be restricted to expressive \ac{dof} while retaining the adaptation flexibility to learn ``realistic'' (\ie actually occurring) shapes of $\bi{\Xi} \left(\cdot,j \right)$.\\
	The key idea of the proposed method is to capture different realizations $\bi{\Xi}\left(\cdot,j\right)$ of the changing system behavior \textit{offline} with a highly flexible \ac{gp} approximation \cite{Solin.2020} and to transform it to a low-dimensional representation. The \ac{gp} approximation lives in $V^{(\phi,N)}$, spanned by a high-dimensional set of basis functions $\{\phi_n \left(\bi{x}_k\right)\}_{n=1}^{N}$. 
	The transformation is done by data-driven extraction of the most significant patterns that account for the change in $\bi{\Xi}\left(\cdot,j\right)$. These patterns are used to condition a new set of few expressive basis functions $\{\rho_m \left(\bi{x}_k\right)\}_{m=1}^{M}$ along which online learning is performed efficiently in a restricted subspace $V^{(\rho,M)} \subset V^{(\phi,N)}$ with less \ac{dof}, \ie $M < N$. Due to the simplicity of the resulting estimation problem, a standard noise-adaptive \ac{pf} \cite{Ozkan.2013} yields sufficient performance. The methodological steps are illustrated in Fig.~\ref{fig:ProposedMethod}.\\ 
	In Sec.~\ref{sec:HilbertGP} and Sec.~\ref{sec:Conditioning}, we consider single-task regression of the $i$-th target function ${\Xi}_i$ and refer to it as $\Xi$ to avoid notation clutter. However, methods for multi-task regression follow trivially by using $n_{\xi}$ single-task regressors in parallel.
	\begin{figure}[t]
		\centering
		\fontsize{8pt}{10pt}\selectfont 
		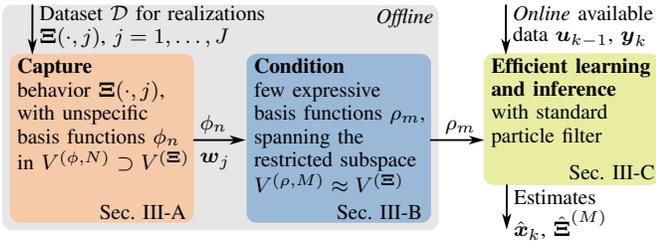
		\normalsize 
		\caption{Proposed method to enable efficient online inference and learning without expert knowledge. First, ``realistic'' (\ie actually occurring) shapes of the target function $\bi{\Xi}(\cdot,j)$, $j=1,\hdots,J$, are captured using the flexible Hilbert-\ac{gp} formulation introduced in \cite{Solin.2020} with unspecific basis functions $\{\phi_n\}_{n=1}^{N}$. Second, in a \textit{data-driven} conditioning step, a new set of few expressive basis functions $\{\rho_m\}_{m=1}^{M}$, $M<N$, is constructed from the most significant patterns in the Hilbert-\ac{gp} coefficients $\bi{w}_j$ without expert knowledge. Based on the obtained low-dimensional approximation, efficient online inference and learning is accomplished using standard \ac{pf} \cite{Ozkan.2013}.}
		\label{fig:ProposedMethod}
	\end{figure}
	\subsection{Capturing the Function Space using Hilbert-\ac{gp}}\label{sec:HilbertGP}
	As a generic representation for learning, we model $\hat{{\Xi}} \sim \mathcal{GP}(0,\kappa(\bi{x}_k,\bi{x}_k'))$ which allows to incorporate prior assumptions regarding $\Xi$, \eg smoothness, intuitively by choosing a kernel $\kappa(\bi{x}_k,\bi{x}_k')$ with corresponding hyperparameters. In particular, we use the \ac{gp} approximation presented in \cite{Solin.2020} due to its beneficial orthogonality properties and its integration in existing inference and learning schemes \cite{Svensson.2017,Berntorp.2021,Berntorp.2022}. The formulation relies on a basis function expansion, and we will refer to it as ``Hilbert-\ac{gp}'' in the following. The concept is briefly revisited along with the presentation of the proposed method. For a detailed introduction, we refer to \cite{Solin.2020}. 
	The main idea is to approximate the kernel using $N$ basis functions ${\phi}_n(\bi{x}_k)$ according to
	\begin{equation}\label{eq:covariance_approx}
		\kappa\left(\bi{x}_k, \bi{x}^{\prime}_k\right) \approx \sum_{n=1}^N S\left( \sqrt{\lambda_n} \right) \phi_n(\bi{x}_k) \phi_n\left(\bi{x}^{\prime}_k\right),
	\end{equation}
	where $S\left( \sqrt{\lambda_n} \right)$ is a factor for encoding the \ac{gp} prior in frequency domain and will be explained later. Using \eqref{eq:covariance_approx}, ${\Xi}(\bi{x}_k, \cdot)$ is approximated by a finite-dimensional basis function expansion
	\begin{equation}\label{eq:Xi_Orig}
			{{\hat{\Xi}}^{(N)}}(\bi{x}_k) =\sum_{n=1}^N w_{n} \phi_{n}(\bi{x}_k) =\bi{w}\tr \bi{\phi}(\bi{x}_k),
	\end{equation}
	with coefficient vectors $\bi{w}\tr = \begin{bmatrix}w_{1} & \hdots & w_{N}\end{bmatrix}$, $\bi{w} \in \mathbb{R}^{N}$, resulting in $n_{\xi}N$ parameters to be found in multi-task regression. The employed basis functions 
	\begin{equation}\label{eq:basisfunctions_phi}
		\phi_{n}(\bi{x}_k) \triangleq \prod_{i=1}^{n_x} \frac{1}{\sqrt{L_i}} \sin \left(\frac{\pi j_i\left(x_{k,i}+L_i\right)}{2 L_i}\right),
	\end{equation}
	are eigenfunctions of the Laplace operator and span the function space $V^{(\phi,N)}$ for input features $\bi{x}_k \in \Omega \subset \mathbb{R}^{n_x}$ on a hypercube domain $\Omega=\left[-L_1,L_1\right] \times \cdots \times \left[-L_{n_x},L_{n_x}\right]$. In the limit $N, L_1, \hdots L_{n_x} \rightarrow \infty$, the basis function expansion converges to the actual \ac{gp} \cite{Solin.2020}. Please note, the eigenfunctions form an orthonormal basis with respect to the inner product $\langle\phi_{n_1}, \phi_{n_2} \rangle$ and have associated eigenvalues
	\begin{equation}\label{eq:eigenvals}
		\lambda_{n} \triangleq \sum_{i=1}^{n_x}\left(\frac{\pi j_i}{2 L_i}\right)^2,
	\end{equation}
	in which each basis function features a unique combination of integers $(j_1, \hdots, j_{n_x})$ that is chosen to maximize the basis functions expressiveness.\\
	The \ac{gp} prior is incorporated by finding a set of weights $w_{j}$ such that the power spectrum of the chosen covariance kernel $\kappa(\bi{x}_k,\bi{x}_k')$ is replicated according to \eqref{eq:covariance_approx}. Here, we employ a squared exponential kernel $\kappa_{\mathrm{se}}(\bi{x}_k,\bi{x}_k')$, as it has been employed successfully for similar settings \cite{Svensson.2017,Berntorp.2021}. The corresponding kernel and spectral density $S_{\mathrm{se}}(\omega)$ are
	\begin{align}
		\kappa_{\mathrm{se}}(\bi{x}_k,\bi{x}_k')&= \sigma^2 \exp \left(-\frac{{\norm{\bi{x}_k-\bi{x}_k'}}^2}{2 l^2}\right),\label{eq:seKernel}\\
		S_{\mathrm{se}}(\omega)&= \sigma^2 \sqrt{2 \pi l^2} \exp \left(-\frac{ l^2 \omega^2}{2}\right), \label{eq:sepowerspectrum}
	\end{align} 
	with hyperparameters $\sigma^2$, $l$ to be defined by the user. Equipped with the basis functions \eqref{eq:basisfunctions_phi}, and the corresponding eigenvalues \eqref{eq:eigenvals}, the varying behavior of target function $\Xi$ is captured by finding suitable coefficient vectors $\bi{w}_j$ for each realization $\Xi\left(\cdot,j\right)$, $j=1,\hdots,J$ in the data set $\mathcal{D}$. This is accomplished by computing the posterior distribution of the coefficients $\bi{w}_j$ \cite{Solin.2020,Menzen.2023} or equivalently solving the regularized least squares problem 
	\begin{equation}
		\bi{w}_j = \argmin_{\bar{\bi{w}}} \sum_{k=1}^{K} \left({\xi}_{k}^{j} - \bar{\bi{w}}\tr \bi{\phi} (\bi{x}_k^{j}) \right)^2  + \sigma_{\xi}^2 \norm{\bar{\bi{w}}}_{\bi{V}^{-1}}^2,
	\end{equation}
	where the \ac{gp} prior is encoded by setting the diagonal regularization matrix $\bi{V}$ with entries $S_{\mathrm{se}}\left( \sqrt{\lambda_n} \right)$, $n = 1, \hdots, N$, following the lines of \cite{Solin.2020}. The combined hyperparameters $\bi{\vartheta} = \{\sigma_{\xi}^2, \sigma^2, l \}$ can be optimized as described in \cite{Solin.2020}.
	
	\subsection{Data-driven Conditioning}\label{sec:Conditioning}
	Having captured the shape of $J$ target function realizations in the coefficient vectors $\bi{w}_j$ of approximation $\hat{\Xi}^{(N)}$, a new set of expressive basis functions is constructed in a \textit{data-driven} fashion. It is worth noting that the target function $\Xi\left(\cdot,j\right)$ usually revisits similar shapes $j \in \left[1,J\right]$ at different time instants in physical systems. As an example, the tire-road friction characteristic is antisymmetric and can be expressed as combinations of $\arctan$ derivates \cite{Pacejka.1992,Berntorp.2022}. Instead of leveraging this by choosing basis functions from expert knowledge, we build a matrix of Hilbert-\ac{gp} parametrizations for the finite set of realizations $\Xi\left(\cdot,j\right)$, $j=1,\hdots,J$, as 
	\begin{equation}
			\bi{W} = \begin{bmatrix}
				\bi{w}_1 & \hdots & \bi{w}_J
			\end{bmatrix}\tr = \bi{U} \bi{\Sigma} \bi{Z}\tr,
	\end{equation}
	and perform a \ac{svd}, yielding the unitary matrices $\bi{U} \in \R^{J \times J}$, $\bi{Z} \in \R^{N \times N}$, and a matrix $\bi{\Sigma} \in \R^{J \times N}$ with singular values $\sigma_j$ as diagonal entries in decreasing order.\\
	Now, the $M$ most significant \ac{dof} can be extracted by choosing the first $M$ columns $\{\bi{z}_m\}_{m=1}^{M}$ of $\bi{Z}$ to define a new set of expressive basis functions $\{\rho_m\left(\bi{x}_k\right)\}_{m=1}^{M}$ with
	\begin{equation}\label{eq:rho}
		{\rho}_m(\bi{x}_k) = \bi{z}_m\tr \bi{\phi}(\bi{x}_k).
	\end{equation}
	The user-defined hyperparameter $M$ is chosen as a trade-off between modeling accuracy and computational complexity of the resulting inference and learning algorithm (which will be described in Sec.~\ref{sec:OnlineLearningAndInference}). The basis functions span the subspace $V^{(\rho,M)}$, restricted to ``realistic'' shapes of the target function $\Xi \left(\cdot,j\right)$. Loosely speaking, each of the basis functions $\rho_m$ corresponds to a specific composition of original basis functions $\phi_n$ and enables learning along a distinct \ac{dof}. Moreover, the new set of basis functions inherits orthogonality properties from vectors $\bi{z}_m$ and original basis functions $\phi_n$.
	\begin{lemma}\label{lem:orthogonality}
		The basis functions $\{ \rho_{m}(\bi{x}_k) \}_{m=1}^{M}$ form an orthonormal set with respect to the inner product $\langle \rho_{i}, \rho_{j} \rangle$, \ie
		\begin{equation}\label{eq:othogonality}
			\int_{\Omega} \rho_{i}(\bi{x}_k) {\rho}_{j}(\bi{x}_k) \mathrm{d}\bi{x}_k = \delta_{ij}.
		\end{equation}
	\end{lemma}
	\begin{proof} 
		See appendix, page \pageref{proof:orthogonality}.
	\end{proof}
	Using the functions $\{ \rho_{m}(\bi{x}) \}_{m=1}^{M}$, a low-dimensional formulation for modeling $\Xi$ can be constructed according to
	\begin{equation}\label{eq:Xi_PCA}
			{{\hat{\Xi}}^{(M)}}(\bi{x}_k) =\sum_{m=1}^M v_{m} \rho_m(\bi{x}_k) =\bi{v}\tr \bi{\rho}(\bi{x}_k), 
	\end{equation}
	in which $\bi{v}\tr = \begin{bmatrix}v_{1} & \hdots & v_{M}\end{bmatrix}$, $\bi{v} \in \mathbb{R}^{M}$ is a coefficient vector with $M < N$. The number of parameters to be determined in multi-task regression is $n_{\xi}M$, independent of the number of original basis functions $N$. Thus, a high-dimensional and accurate Hilbert-\ac{gp} formulation can be tailored to learn $\Xi (\cdot,j)$ online with few \ac{dof}, significantly reducing the complexity of the inference problem.\\
	In Fig.~\ref{fig:BasisFunctionAccuracy}, the approximation accuracy of both basis function expansions \eqref{eq:Xi_Orig} and \eqref{eq:Xi_PCA} in a numerical example is shown for a set of realizations $\Xi(\cdot,j)$, $j=1,\hdots,J$, depending on the number of parameters to be learned online. The results suggest that significantly fewer \ac{dof} are required with the proposed approach. In particular, the accuracy of the original Hilbert-\ac{gp} increases step-wise with every second \ac{dof}, indicating that basis functions with even integers $j_i$ in \eqref{eq:basisfunctions_phi} are inexpressive to represent the chosen target function due to symmetry. In \cite{Berntorp.2022}, this is exploited to facilitate learning based on expert knowledge. In contrast, the proposed method exploits the \ac{dof} effectively in a \textit{data-driven} fashion.\\
	Moreover, the distance between $\hat{\Xi}^{(N)}(\bi{x}_k)$ and $\hat{\Xi}^{(M)}(\bi{x}_k)$ can be quantified prior to evaluation.
	\begin{figure}[t]
		\centering
		\includegraphics[width=0.49\textwidth]{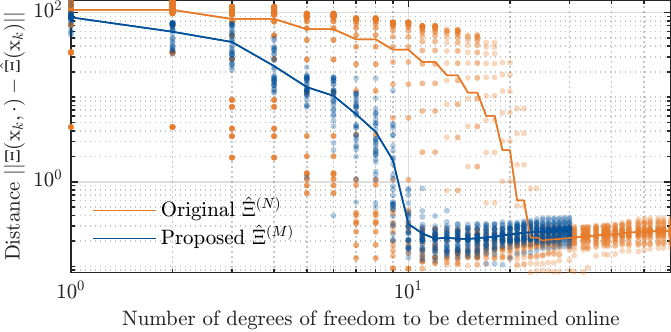}
		\caption{Error between the true function $\Xi(x_k,j) = 10 \text{sinc}(j x_k /100)$ and the basis function expansions $\hat{\Xi}^{(i)}(x_k)$ for $i=N,M$ and for different realizations $j=1,\hdots,30$ on the domain $\Omega = \left[-15,15\right]$. Each error result corresponds to a dot, and the mean is drawn as a line. The number of required \ac{dof} to achieve a certain approximation performance is significantly reduced using the proposed approach due to the choice of expressive basis functions.}
		\label{fig:BasisFunctionAccuracy}
	\end{figure}
	\begin{theorem}\label{theo:approximation_accuracy}
		The distance $d(\bi{x}_k) = {\hat{\Xi}}^{(N)}(\bi{x}_k) - {\hat{\Xi}}^{(M)}(\bi{x}_k)$ between the basis function expansions ${\hat{\Xi}}^{(N)}(\bi{x}_k)$ and ${\hat{\Xi}}^{(M)}(\bi{x}_k)$ on the domain $\Omega$ is given by
		\begin{equation}
			{\norm{d(\bi{x}_k)}}^2 = {\norm{\bi{w} - \sum_{m=1}^{M} {v}_{m} \bi{z}_m }}^2.
		\end{equation}
	\end{theorem}
	\begin{proof} 
		See appendix, page \pageref{app:approximation_accuracy}.
	\end{proof}
	\begin{remark}
		If the target functions in multi-task regression $\Xi_i(\cdot,j)$, $i=1,...,n_\xi$, are of significantly different shape, the number of basis functions needs to be increased to capture relevant characteristics. Alternatively, separate Hilbert-\ac{gp}s and/or conditioning steps per target can be employed.
	\end{remark}
	\begin{remark}\label{rem:approx_error1}
		If we choose $v_{m}~=~\sum_{m=1}^{M} {u}_{jm} \sigma_m$ and $M=\min\left(J,N\right)$ for modeling realization $j$, with ${u}_{jm}$ and $\sigma_m$ being elements of $\bi{U}$ and $\bi{\Sigma}$, the distance ${\norm{d(\bi{x}_k)}}^2 = 0$, and the original basis function expansion is recovered. 
	\end{remark}
	\begin{remark}\label{rem:approx_error2}
		In practice, we see that the approximation error of the reduced-order representation with respect to the true function does not decrease asymptotically towards the original basis function expansion error as we add \ac{dof} (see Fig.~\ref{fig:BasisFunctionAccuracy}). Instead, ${\hat{\Xi}}^{(M)}$ usually reaches the highest accuracy for $M < J$, \ie if only the most significant \ac{dof} are used.
	\end{remark}
	\begin{remark}\label{rem:generalization}
		As the underlying \textit{dominant features} in the offline data set are extracted, the learnable function space $V^{(\rho,M)}$ is bound to linear combinations of these dominant features. Therefore, $V^{(\rho,M)}$ can contain function shapes (\ie linear feature combinations) that are not present in the offline data set, with potential implications for generalization beyond the training data distribution.
	\end{remark}	
	For online inference and learning of a Hilbert-\ac{gp}, specifically tailored \ac{smc} have been proposed in \cite{Berntorp.2021,Berntorp.2022}. In contrast, the low-dimensional basis function expansion $\hat{\Xi}^{(M)}$ simplifies the estimation problem significantly, such that a standard \ac{pf} \cite{Ozkan.2013} can be used here.
	
	\subsection{Efficient Online Inference and Learning}\label{sec:OnlineLearningAndInference}
	Having obtained the new set of orthonormal basis functions $\{\rho_m\left(\bi{x}_k\right)\}_{m=1}^{M}$, efficient learning of the target functions ${\Xi}_i(\cdot,j)$, $i=1,\hdots,n_{\xi}$, in the restricted subspace $V^{(\rho,M)}$ is possible using standard Bayesian inference methods. 
	To this end, we formulate a system model that contains the parameters $\bi{v}_i$ to be learned online following an ad-hoc state augmentation approach, as done in a related setting \cite{Kullberg.2021c}. The reason for this is two-fold:\\
	(i) Parameter estimation by state augmentation is a baseline approach that, despite its simplicity, yields sufficient estimation quality for a wide range of applications and is thus commonly applied and accepted in practice.\\
	(ii) The target function $\bi{\Xi}$ is nonlinearly \textit{nested} inside known system dynamics $\bi{f}$. Due to the nonlinearly nested structure \eqref{eq:model_structure} in the present setting, the inverse model $\bi{\Xi}(\bi{x}_k, j) = \left. \bi{f}^{-1} (\bi{x}_{k+1}) \right\lvert_{\bi{x}_{k},\bi{u}_{k}}$ is not generally known for updating the posterior of $\bi{v}_i$, formally obstructing direct application of \cite{Svensson.2017,Berntorp.2021,Berntorp.2022}.\\
	In this light, we model the system \eqref{eq:model_structure} using the basis function expansion $\hat{\bi{\Xi}}^{(M)}$ (restricted to the subspace $V^{(\rho,M)}$), by
	\begin{subequations}\label{eq:estimation_model}
		\begin{align}
			\tilde{\bi{x}}_{k+1} &= \bi{F} \left(\tilde{\bi{x}}_k, \bi{u}_k\right) + \tilde{\bi{e}}_k\\
			&= \begin{bmatrix} \bi{x}_{k+1} \\ \bi{v}_{1,k+1} \\ \vdots \\ \bi{v}_{n_{\xi},k+1} \end{bmatrix} = \begin{bmatrix} \bi{f}(\bi{x}_{k}, \bi{u}_{k}, \hat{\bi{\Xi}}^{(M)}(\bi{x}_{k})) \\ \bi{v}_{1,k} \\ \vdots \\ \bi{v}_{n_{\xi},k} \end{bmatrix} +  \begin{bmatrix} \bi{e}_{k}^{x} \\ \bi{e}_{k}^{v_1} \\ \vdots \\ \bi{e}_{k}^{{v_{n_\xi}}} \end{bmatrix},\nonumber\\
			\bi{y}_{k} &=\bi{{h}}({\bi{x}}_{k}, \bi{u}_k) + \bi{e}_{k}^{y},
		\end{align}
	\end{subequations}
	with a random walk assumption on the parameters $\bi{v}_i$. The parameter noise for target function $\Xi_i$ is shaped according to the significance of the respective \ac{dof} (\ie the basis function $\rho_m(\bi{x}_k)$), represented by the truncated singular value matrix $\bi{\Sigma}_i$ from the conditioning step. Hence, the parameter process noise is drawn $\bi{e}_k^{v_i} \sim \N \left(\bi{0}, c \bi{\Sigma}_i \right)$, with $c$ being a user-defined design parameter that scales the exploration capability. Hence, the overall process noise has a block-diagonal covariance $\tilde{\bi{Q}}$ and is drawn ${\bi{e}}_k^{\tilde{x}} \sim \N ( \bi{0}, \tilde{\bi{Q}})$.\\
	Based on the model formulation \eqref{eq:estimation_model}, the noise-adaptive marginalized particle filter proposed in \cite{Ozkan.2013} is employed. The motivation for noise adaptation is that, if the true function $\bi{\Xi}$ changes rapidly, the error between measured evidence and the current estimate of the system state increases. In other words, the current approximation $\prob\left({\tilde{\bi{x}}}_k \vert {\tilde{\bi{x}}}_{0: k-1}, \bi{y}_{0: k-1}\right)$ does not represent new evidence $\bi{y}_k$ well which would lead to rapid algorithm divergence without noise adaptivity.\\
	In particular, we model the measurement noise covariance $\bi{R}_k$ with an inverse Wishart distribution $\bi{R}_k \sim \IW (\nu_k,\bi{\Lambda}_k)$, which is the conjugate prior for the multivariate normal distribution. In principle, the process noise of the parameters could be adapted as well, but this lead to frequent path degeneracy in simulative case studies. For self-contained presentation, the main steps of the noise-adaptive particle filter in \cite{Ozkan.2013} are presented and connected to the problem at hand subsequently. For notational clarity, dependence on the inputs $\bi{u}_k$ is omitted.\\
	The overall target is the joint probability density function of state trajectory $\tilde{\bi{x}}_{0:k}$ and current noise parameters $\bi{\theta}_k~=~\bi{R}_k$, given the measurements $\bi{y}_{0:k}$
	\begin{equation}
		\prob\left( \tilde{\bi{x}}_{0:k}, \bi{\theta}_k \vert \bi{y}_{0:k} \right) = \underbrace{\prob\left( \bi{\theta}_k \vert \tilde{\bi{x}}_{0:k}, \bi{y}_{0:k} \right)}_{\text{Posterior (II)}} \underbrace{\prob\left( \tilde{\bi{x}}_{0:k} \vert \bi{y}_{0:k} \right)}_{\text{Posterior (I)}},
	\end{equation}
	that is composed of the recursively computed posterior density of the states $\prob\left( \tilde{\bi{x}}_{0:k} \vert \bi{y}_{0:k} \right)$ and the posterior density $\prob\left( \bi{\theta}_k \vert \tilde{\bi{x}}_{0:k}, \bi{y}_{0:k} \right)$ of the noise parameters. The posterior (I) is approximated by a set of $N_p$ weighted particles
	\begin{equation}
		\prob\left(\tilde{\bi{x}}_{0: k} \vert \bi{y}_{0: k}\right) \approx \sum_{i=1}^{N_p} q_k^i \delta_{\tilde{\bi{x}}_{0: k}^i, \tilde{\bi{x}}_{0: k}^{~}},
	\end{equation}
	that represent different state trajectories. The weights $q_k^i$ capture the probability of the respective trajectory at the current time instant $k$ and are recursively updated as
	\begin{equation}\label{eq:weightupdate}
		{q}_k^i \propto q_{k-1}^i \frac{\prob\left(\bi{y}_k\vert \tilde{\bi{x}}_{0: k}^i, \bi{y}_{0: k-1}\right)\prob\left(\tilde{\bi{x}}_k^i \vert \tilde{\bi{x}}_{0: k-1}^i, \bi{y}_{0: k-1}\right)}{\pi\left(\tilde{\bi{x}}_k^i \vert \tilde{\bi{x}}_{0: k-1}^i, \bi{y}_{0: k}\right)},
	\end{equation}
	when new measurement evidence $\bi{y}_k$ becomes available. The density $\pi\left(\tilde{\bi{x}}_k^i \vert \tilde{\bi{x}}_{0: k-1}^i, \bi{y}_{0: k}\right)$ is a tractable proposal distribution from which the states are drawn \cite{Sarkka.2023}. The employed likelihood $\prob\left( \bi{y}_{k} \vert  \tilde{\bi{x}}_{0:k}^i, \bi{y}_{0:k-1} \right)$ is obtained by integrating out the noise parameters particle-based, according to
	\begin{equation}
		\begin{aligned}
			&\prob\left( \bi{y}_{k} \vert  \tilde{\bi{x}}_{0:k}, \bi{y}_{0:k-1} \right) = \int \prob\left(\bi{y}_{k} \vert  \tilde{\bi{x}}_{k}, \bi{\theta}_{k-1} \right) \\
			& \hspace{2.5cm} \times \prob\left(\bi{\theta}_{k-1} \vert  \tilde{\bi{x}}_{0:k}, \bi{y}_{0:k-1} \right) \mathrm{d} \bi{\theta}_{k-1},
		\end{aligned}
	\end{equation}
	which, due to the inverse Wishart prior, is a Student-t ($\mathcal{T}$) distribution
	\begin{equation}\label{eq:studentt}
			\prob\left(\bi{y}_k \vert {\tilde{\bi{x}}}_{k}^i\right)
			=\T \left(\bi{h}\left({{\bi{x}}}_k^{i}, \bi{u}_k\right), \bi{\Lambda}_k, \nu_k - n_y + 1 \right),
	\end{equation}
	for each particle, dependent on its current noise statistics. As stated earlier, if the underlying function $\bi{\Xi}$ changes suddenly, the current approximation $\prob\left({\tilde{\bi{x}}}_k^i \vert {\tilde{\bi{x}}}_{0: k-1}^i, \bi{y}_{0: k-1}\right)$ does not represent new evidence $\bi{y}_k$ well, leading to rapid algorithm divergence without noise adaptivity. The inverse Wishart prior on $\bi{R}_k$ accounts for this by adapting the noise covariance, effectively exploring a larger search space.\\
	Given the state estimates $\tilde{\bi{x}}_k^i$, the posterior density (II) can be evaluated. As the parameter posterior is again an inverse Wishart distribution $\IW (\nu_k,\bi{\Lambda}_k)$, this amounts to updating the parameter statistics with the ``measurement'' $\bi{p}_k = \bi{y}_k - \bi{h}(\tilde{\bi{x}}_k)$ according to
	\begin{subequations}\label{eq:statistics_update}
		\begin{align}
			\nu_{k \lvert k} &=  \nu_{k \lvert k-1} +1,\\		
			\bi{\Lambda}_{k \lvert k} &=  \bi{\Lambda}_{k \lvert k-1} + \bi{p}_k \bi{p}_k\tr.
		\end{align}
	\end{subequations}
	The initial covariance is sampled $\bi{R}_0 \sim \IW\left(\nu_0, \bi{\Lambda}_0\right)$. If the noise parameters are time-varying, a forgetting factor $\lambda_{\mathrm{f}}$ can be incorporated in the prediction step of the statistics to reduce the impact of old observations and introduce mixing \cite{Berntorp.2021,Sarkka.2023}. The resulting method is implemented as a bootstrap \ac{pf} and summarized in Algorithm \ref{alg:OnlineInferenceAndLearning}.
	\begin{algorithm}[H]
		\caption{Pseudo-code of the proposed algorithm}
		\label{alg:OnlineInferenceAndLearning}
		\textbf{Data-driven conditioning:} Compute expressive basis functions $\{\rho_m(\bi{x}_k) \}_{m=1}^{M}$ according to \eqref{eq:rho} and choose initial coefficients $\{{\bi{v}}_{j,0}\}_{j=1}^{n_\xi}$.\\
		\textbf{Initialize:} Set $\{\tilde{\bi{x}}_{0}^{i}\}_{i=1}^{N_p} \sim p(\tilde{\bi{x}}_0)$, $\{\bi{\Lambda}_0^{i}, \nu_0^{i}\}_{i=1}^{N_p} = \{\bi{\Lambda}_0, \nu_0\}$, $\lambda_{\mathrm{f}} \in (0,1]$.
		\begin{algorithmic}[1]
			\For{$k = 1, \hdots$}
			\State Read current data $\bi{u}_{k-1}$, $\bi{y}_k$.
			\For{$i = 1, \hdots, N_p $}
			\State Time update of noise statistics 
			\Statex \qquad \quad ${\nu}_{k \lvert k-1}^{i} = \lambda_{\mathrm{f}} {\nu}_{k-1 \lvert k-1}^{i}$, $\bi{\Lambda}_{k \lvert k-1}^{i} = \lambda_{\mathrm{f}} \bi{\Lambda}_{k-1 \lvert k-1}^{i}$. 
			\State Sample $\tilde{\bi{x}}_{k}^{i} \sim \N(\bi{F} (\tilde{\bi{x}}_{k-1}^{i},\bi{u}_{k-1}), \tilde{\bi{Q}})$.
			\State Compute weight $\bar{q}_{k}^{i} = \prob\left(\bi{y}_k \vert {\bi{x}}_{k}^i\right)$ using \eqref{eq:studentt}.
			\State Measurement update of noise statistics \eqref{eq:statistics_update}.
			\EndFor
			\State Normalize weights $q_{k}^{i} = {\bar{q}_k^{i}}/{\sum_{i=1}^{N_p} \bar{q}_k^{i}}$.
			\State Compute estimates $ \hat{{\bi{x}}}_k, \hat{\bi{v}}_{1,k}, \hdots, \hat{\bi{v}}_{n_\xi,k}$.
			\State Resample particles and copy the  
			\Statex \hspace{4mm} corresponding noise statistics.
			\EndFor
		\end{algorithmic}
	\end{algorithm}
	\begin{figure*}[t]
		\centering
		\fontsize{8pt}{10pt}\selectfont 
\begingroup%
  \makeatletter%
  \providecommand\color[2][]{%
    \errmessage{(Inkscape) Color is used for the text in Inkscape, but the package 'color.sty' is not loaded}%
    \renewcommand\color[2][]{}%
  }%
  \providecommand\transparent[1]{%
    \errmessage{(Inkscape) Transparency is used (non-zero) for the text in Inkscape, but the package 'transparent.sty' is not loaded}%
    \renewcommand\transparent[1]{}%
  }%
  \providecommand\rotatebox[2]{#2}%
  \newcommand*\fsize{\dimexpr\f@size pt\relax}%
  \newcommand*\lineheight[1]{\fontsize{\fsize}{#1\fsize}\selectfont}%
  \ifx\svgwidth\undefined%
    \setlength{\unitlength}{505.89001465bp}%
    \ifx\svgscale\undefined%
      \relax%
    \else%
      \setlength{\unitlength}{\unitlength * \real{\svgscale}}%
    \fi%
  \else%
    \setlength{\unitlength}{\svgwidth}%
  \fi%
  \global\let\svgwidth\undefined%
  \global\let\svgscale\undefined%
  \makeatother%
  \begin{picture}(1,0.43487713)%
    \lineheight{1}%
    \setlength\tabcolsep{0pt}%
    \put(0,0){\includegraphics[width=\unitlength,page=1]{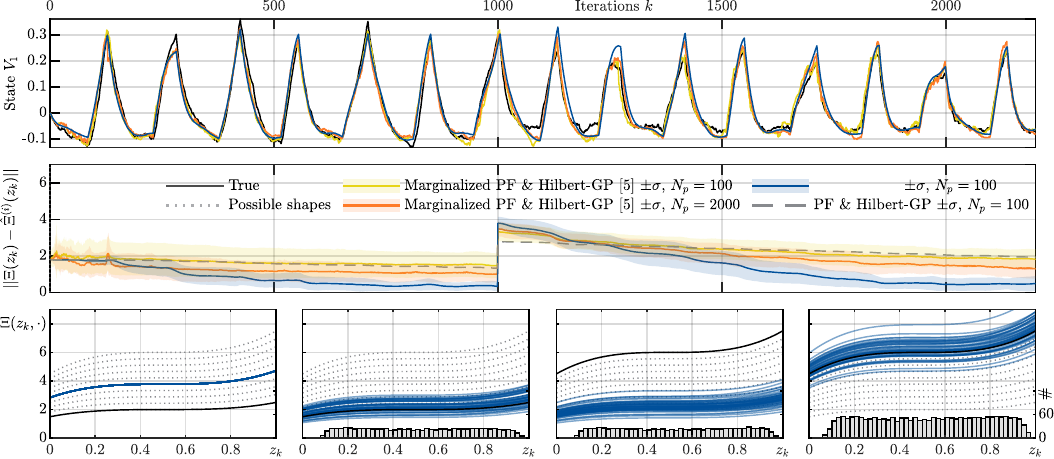}}%
    \put(0.99627605,0.07004866){\color[rgb]{0,0,0}\rotatebox{90}{\makebox(0,0)[lt]{\lineheight{1.10000002}\smash{\begin{tabular}[t]{l}Data Points\end{tabular}}}}}%
    \put(0.77233734,0.25414157){\color[rgb]{0,0,0}\makebox(0,0)[lt]{\lineheight{1.10000002}\smash{\begin{tabular}[t]{l}Algorithm \ref{alg:OnlineInferenceAndLearning}\end{tabular}}}}%
    \put(0,0){\includegraphics[width=\unitlength,page=2]{CDC_fig3_final.pdf}}%
  \end{picture}%
\endgroup%

		\normalsize 
		\caption{Online inference and learning in the simulation example \eqref{eq:batterymodel} using the proposed Algorithm \ref{alg:OnlineInferenceAndLearning} for learning with the expressive basis function expansion $\hat{\Xi}^{(M)}$ and the \ac{pf} described in Sec.~\ref{sec:OnlineLearningAndInference}. Algorithm \ref{alg:OnlineInferenceAndLearning} converges rapidly from a wrong initial condition and after a sudden change in the true function $\Xi$ at $k=1,000$. The shown error is the mean for $50$ Monte-Carlo runs (standard deviation $\sigma$ shown semi-transparent), and the estimation results $\hat{\Xi}^{(M)}$ for the respective Monte-Carlo runs are presented in the bottom plots for the indicated iterations.}
		\label{fig:OnlineLearningAndInference}
	\end{figure*}
	\section{SIMULATION RESULTS}\label{sec:results}
	For evaluation of Algorithm \ref{alg:OnlineInferenceAndLearning}, a numerical simulation example inspired by the nonlinear battery model in \cite{Aitio.2023} is used. The continuous-time state-space model 
	\begin{subequations}\label{eq:batterymodel}
		\begin{align}
			\dot{\bi{x}} & = \begin{bmatrix}
				\dot{z} \\ \dot{V}_{1} \\ \dot{T}_{c}
			\end{bmatrix} =  \begin{bmatrix}
				I Q_{\mathrm{bat}}^{-1} \\
				-\alpha \left( z, j \right) V_{1} + \beta \left( z \right) I\\
				\frac{1}{C_{\mathrm{c}}} \left( V_{1} I  + R_0 \left( z, I \right) I^2 -\frac{T_{\mathrm{c}} - T_{\mathrm{a}} }{ R_{\mathrm{c}} }\right) 
			\end{bmatrix},\\
			\bi{y} &= \begin{bmatrix}
				z, &
				V_0 \left( z \right) + V_{1} + R_0\left( z, I \right) I, &
				T_{\mathrm{c}}
			\end{bmatrix}\tr,
		\end{align}
	\end{subequations}
	exhibits $3$ states and $3$ outputs. The ``charging current'' $I$ is a scalar input. A discrete-time representation is obtained by $4$-th order Runge-Kutta integration (RK4) with $0.01$ seconds step width. For simulation, additive process and measurement noise $\bi{e}_k^{x}\sim \N(\bi{0}, 10^{-5} \times \mathbf{I}_3)$ and $\bi{e}_k^{y}\sim \N(\bi{0},10^{-2} \times \mathbf{I}_3)$ is considered, respectively. The nested function $\alpha(z_k,j)$ represents a nonlinear and parametric relationship related to an RC circuit that varies over time and is to be learned online. Thus, $\Xi(x_{1},j)=\alpha(z_k,j)$ and the true function $\alpha(z_k,j) = 4 j - 8 j \left( 0.5 - z_k \right)^3$ for $j = 1, \hdots, 10$. Please note, that $\alpha$ is nested nonlinearly in the discrete-time system dynamics due to RK4 discretization. The remaining quantities are physical parameters and nonlinear functions. For further details, the reader is referred to \cite{Aitio.2023}. In the evaluation scenario, a wrong initialization of $\hat{\Xi}$ is set, and at time step $k=1,000$, a sudden change from $\alpha (\cdot,1)$ to $\alpha (\cdot,10)$ is simulated. In both cases, the true function shape should be found.\\
	To initialize Algorithm \ref{alg:OnlineInferenceAndLearning}, we condition $M = 2$ expressive basis functions $\rho_{m}$ based on coefficients for $N=50$ original basis functions $\phi_n$, leveraging that the number of \ac{dof} to be learned online is independent of $N$. The user-defined parameters $\bi{\Lambda}_0 = \mathbf{I}_3$, $\nu_0 = 3$, $c = 3 \times 10^{-5}$, and the number of particles $N_p = 100$.\\ 
	For comparison, we first use, as a baseline approach, a Hilbert-\ac{gp} combined with the \ac{pf} described in Sec.~\ref{sec:OnlineLearningAndInference}. Second, as a state-of-the-art online inference and learning method, we use the tailored marginalized \ac{pf} for online learning of a Hilbert-\ac{gp} state-space model proposed in \cite{Berntorp.2021}. All Hilbert-\ac{gp} representations in the comparison study are pre-trained and initialized with coefficient vectors corresponding to the function realization $\alpha(\cdot,5)$ (\ie all methods start with the same initial conditions). The hyperparameters $\sigma^2$, $l$ and $\sigma_\xi^2$ are optimized according to \cite{Solin.2020} using the offline data set $\mathcal{D}$. The remaining design parameters in the comparison methods are hand-tuned.\\ 
	In Fig.~\ref{fig:OnlineLearningAndInference}, true values and estimates for state $V_1$ and target function $\alpha$ are shown in the top and bottom plots, respectively. In the middle plot, the estimation error regarding the target function is depicted. Despite the use of up to $2,000$ particles, the baseline, and the state-of-the-art comparison both show slower convergence than the proposed approach, which can be attributed to the complex inference task of determining the parameters of a highly flexible, unconditioned Hilbert-\ac{gp}. In contrast, Algorithm \ref{alg:OnlineInferenceAndLearning} converges rapidly to the true function $\alpha$ both, from a wrong initialization and after a sudden change of the system behavior. The learning effect is visible in the estimates of $V_1$ as well. The overall state estimate accuracy is comparable across the different methods in the present example simulation.\\	
	The execution of the unoptimized code took on average $5\,\mathrm{ms}$ (\ac{pf} \& Hilbert-\ac{gp}, $N_p=100$), $3\,\mathrm{ms}$ (\cite{Berntorp.2021}, $N_p=100$), $57\,\mathrm{ms}$ (\cite{Berntorp.2021}, $N_p=2,000$), and $7\,\mathrm{ms}$ (Algorithm \ref{alg:OnlineInferenceAndLearning}, $N_p=100$) per time step, respectively (i5-1235U CPU, 8 GB RAM). The results indicate that Algorithm \ref{alg:OnlineInferenceAndLearning} converges faster to the target function than a state-of-the-art method which uses significantly more particles and computational resources. In this light, Algorithm \ref{alg:OnlineInferenceAndLearning} can be considered computationally efficient thanks to the reduced number of required particles.
	
	\acresetall
	\section{CONCLUSIONS}\label{sec:conclusion}
	In the current work, a data-driven offline conditioning step for efficient online inference of latent states and learning of an unknown target function is proposed. The key idea is to restrict learning to a low-dimensional subspace spanned by expressive \acp{dof} without prior expert knowledge about the target function. In operation, {online} inference and learning is performed {efficiently} along these \ac{dof}. Compared to a baseline method and a state-of-the-art method, the proposed approach yields a significantly simplified estimation problem without relying on expert knowledge about the target function. Moreover, the scheme is capable of learning a nonlinearly \textit{nested} target function {inside} a first-principles model, which is addressed only in few existing works on online inference and learning.\\ 
	Thus, we contribute a method that has the potential to facilitate the operation of real-world systems in complex and changing conditions. Specifically, the proposed scheme provides a further step towards intelligent operation under fluctuating resistance forces, changing geometries, and parameters due to varying temperature, location, or time. Relevant real-world effects include wear and aging.\\
	In future research, the proposed conditioning step to extract expressive basis functions might be integrated into other online inference and learning schemes, \eg \cite{Berntorp.2021}, to facilitate learning. Another interesting research direction is to incorporate conditioning directly in offline particle Markov chain Monte Carlo methods to obtain a set of expressive basis functions for online inference and learning.
	
	\section*{ACKNOWLEDGMENT}
	For implementation of the proposed methods, parts of the program code of \cite{Kok.2024} have been used and modified. Jan-Hendrik Ewering would like to thank Prof. Manon Kok for making the implementation of \cite{Kok.2024} publicly available.
	
	\balance
	\bibliographystyle{IEEEtran}
	\bibliography{references}
	
	\newpage
	\section*{APPENDIX}
	\subsection{Proof of Lemma \ref{lem:orthogonality}}\label{proof:orthogonality}
	For notational convenience, we omit the discrete-time index $k$ in this proof. To show the orthogonality of the basis functions $\{ \rho_i (\bi{x}) \}_{i=1}^M$, we start by inserting the definitions of the basis functions in \eqref{eq:othogonality}, noting that the functions are real-valued and expanding the products, which yields
	\begin{align}
		\begin{aligned}
			&\langle \rho_i, \rho_j \rangle 
			= \int_{\Omega} \rho_i(\bi{x}) {\rho}_j(\bi{x}) \mathrm{d}\bi{x}\\
			&= \int_{\Omega} \left(\bi{z}_i\tr\bi{\phi}(\bi{x}) \right) \left( \bi{z}_j\tr\bi{\phi}(\bi{x}) \right) \mathrm{d}\bi{x}. \\
		\end{aligned}
	\end{align}
	Expanding the product into two sums and rearranging gives
	\begin{align}
		\begin{aligned}
			&\langle \rho_i, \rho_j \rangle 
			= \int_{\Omega} \rho_i(\bi{x}) {\rho}_j(\bi{x}) \mathrm{d}\bi{x}\\
			&= \int_{\Omega} \left(\bi{z}_i\tr\bi{\phi}(\bi{x}) \right) \left( \bi{z}_j\tr\bi{\phi}(\bi{x}) \right) \mathrm{d}\bi{x} \\
			&= \int_{\Omega} \left(\sum_{n_1=1}^N z_{i n_1} \phi_{n_1}(\bi{x}) \right) \left( \sum_{n_2=1}^N z_{j n_2} \phi_{n_2}(\bi{x}) \right) \mathrm{d}\bi{x}\\
			&= \int_{\Omega} \sum_{n_1=1}^N  \sum_{n_2=1}^N z_{i n_1} z_{j n_2} \phi_{n_1}(\bi{x})   \phi_{n_2}(\bi{x}) \mathrm{d}\bi{x}.
		\end{aligned}
	\end{align}
	The order of the integral(s) and the sums can be changed because the sums are uniformly convergent. This can be seen, if we reorder the summands by decreasing order of $z_{i n_1} z_{j n_2}$ and extend the finite sums to an infinite series by adding zeros for $n_1 n_2 > N^2$. In this case, uniform convergence is provided by Dirichlet's test for uniform convergence. 
	\begin{align}\label{eq:dirichlet}
		\begin{aligned}
			&\langle \rho_i, \rho_j \rangle 
			= \sum_{n_1=1}^N  \sum_{n_2=1}^N z_{i n_1} z_{j n_2} \int_{\Omega}  \phi_{n_1}(\bi{x})   \phi_{n_2}(\bi{x}) \mathrm{d}\bi{x}\\
			&= \sum_{n_1=1}^N  \sum_{n_2=1}^N z_{i n_1} z_{j n_2} \delta_{n_1 n_2}
			= \bi{z}_i\tr \bi{z}_j = \delta_{ij},\\
		\end{aligned}
	\end{align}
	because the vectors $\bi{z}_i$ and $\bi{z}_j$ are columns of $\bi{Z}$ and form an orthonormal basis.
	
	\subsection{Proof of Theorem \ref{theo:approximation_accuracy}}\label{app:approximation_accuracy}
	For notational convenience, we omit the index $k$ in this proof and define
	\begin{align}
		f^{(N)}& \triangleq {\hat{\Xi}}^{(N)}(\bi{x}), \quad
		f^{(M)} \triangleq {\hat{\Xi}}^{(M)}(\bi{x}).
	\end{align}
	The squared L$2$ norm of the distance function
	\begin{align}\label{eq:error}
		\begin{aligned}
			&{\norm{d(\bi{x})}}^2
			= \int_{\Omega} \left(f^{(N)}-f^{(M)}\right) \left(f^{(N)}-f^{(M)}\right) \mathrm{d}\bi{x}\\
			&= \underbrace{\int_{\Omega} f^{(N)}f^{(N)} \mathrm{d}\bi{x}}_{\text{Term 1}} -2 \underbrace{\int_{\Omega} f^{(N)}f^{(M)} \mathrm{d}\bi{x}}_{\text{Term 2}} + \underbrace{ \int_{\Omega} f^{(M)}f^{(M)} \mathrm{d}\bi{x}}_{\text{Term 3}}.
		\end{aligned}
	\end{align}
	The integral is decomposed for the three terms and each is considered separately. Following the same argumentation as in \eqref{eq:dirichlet}, term 1 yields
	\begin{align}
		\begin{aligned}
			&\int_{\Omega} f^{(N)}f^{(N)} \mathrm{d}\bi{x}\\
			&= \int_{\Omega} \left(\bi{w}\tr \bi{\phi}(\bi{x}) \right) \left( \bi{w}\tr \bi{\phi}(\bi{x}) \right) \mathrm{d}\bi{x} \\
			&= \sum_{n_1=1}^N \sum_{n_2=1}^N w_{n_1} w_{n_2} \int_{\Omega}  \phi_{n_1}(\bi{x})   \phi_{n_2}(\bi{x}) \mathrm{d}\bi{x}\\
			&= \sum_{n_1=1}^N \sum_{n_2=1}^N w_{n_1} w_{n_2} \delta_{n_1 n_2}
			= \bi{w}\tr \bi{w}.
		\end{aligned}
	\end{align}
	Similarly, using the argumentation of the proof of Lemma \ref{lem:orthogonality}, term 3 gives
	\begin{align}
		\begin{aligned}
			&\int_{\Omega} f^{(M)}f^{(M)} \mathrm{d}\bi{x} = \bi{v}\tr \bi{v}.
		\end{aligned}
	\end{align}
	Term 2
	\begin{align}
		\begin{aligned}
			&\int_{\Omega} f^{(N)}f^{(M)} \mathrm{d}\bi{x}\\
			&= \int_{\Omega} \left(\bi{w}\tr \bi{\phi}(\bi{x}) \right) \left( \bi{v}\tr \bi{\rho}(\bi{x}) \right) \mathrm{d}\bi{x} \\
			&=  \sum_{m=1}^M \sum_{n=1}^N  v_{m} w_{n} \int_{\Omega}  \phi_{n}(\bi{x})   \rho_{m}(\bi{x}) \mathrm{d}\bi{x},
		\end{aligned}
	\end{align}
	with $\rho_{m}(\bi{x})$ given by \eqref{eq:rho}. Using the definition, we obtain
	\begin{align}
		\begin{aligned}
			&\int_{\Omega} f^{(N)}f^{(M)} \mathrm{d}\bi{x}\\
			&= \sum_{m=1}^M \sum_{n_1=1}^N   v_{m} w_{n_1} \int_{\Omega}  \phi_{n_1}(\bi{x})   \sum_{n_2=1}^N z_{n_2m} \phi_{n_2}(\bi{x})  \mathrm{d}\bi{x}\\
			&= \sum_{m=1}^M \sum_{n_1=1}^N   v_{m} w_{n_1} \int_{\Omega}  \sum_{n_2=1}^N z_{n_2m} \phi_{n_1}(\bi{x}) \phi_{n_2}(\bi{x})  \mathrm{d}\bi{x}\\
			&= \sum_{m=1}^M \sum_{n_1=1}^N   v_{m} w_{n_1}   \sum_{n_2=1}^N z_{n_2m} \int_{\Omega} \phi_{n_1}(\bi{x}) \phi_{n_2}(\bi{x})  \mathrm{d}\bi{x}\\
			&= \sum_{m=1}^M \sum_{n_1=1}^N   v_{m} w_{n_1}   \sum_{n_2=1}^N z_{n_2m} \delta_{n_1 n_2}\\
			&= \sum_{m=1}^M \sum_{n=1}^N   v_{m} w_{n}  z_{nm}
			= \sum_{m=1}^M   v_{m}   \bi{z}_m\tr \bi{w}
		\end{aligned}
	\end{align}
	where, in the second step, a similar argumentation as in the proof of Lemma \ref{lem:orthogonality} is considered for changing the order of the sum and the integral. 
	Plugging the terms back into \eqref{eq:error} gives
	\begin{align}
		\begin{aligned}
			&{\norm{d(\bi{x})}}^2
			= \bi{w}\tr \bi{w} - 2 \sum_{m=1}^M   v_{m}   \bi{z}_m\tr \bi{w} + \bi{v}\tr \bi{v}\\
			&= \bi{w}\tr \bi{w} - 2 \sum_{m=1}^M   v_{m}   \bi{z}_m\tr \bi{w} + \sum_{m=1}^M   v_{m}   \bi{z}_m\tr \bi{z}_m v_{m}\\
			&= {\norm{\bi{w} - \sum_{m=1}^M   v_{m}   \bi{z}_m}}^2.
		\end{aligned}
	\end{align}
\end{document}